\documentclass[conference]{IEEEtran}

\usepackage[cmex10]{amsmath}
\usepackage{amssymb}
\usepackage{cite}
\usepackage{amsmath}
\usepackage{wasysym}
\usepackage{booktabs}
\usepackage{paralist}
\usepackage{balance}
\usepackage{subfigure}
\usepackage{multicol}
\usepackage{graphicx}
\usepackage{epsfig}
\usepackage{epstopdf}
\usepackage[ruled,linesnumbered]{algorithm2e}
\usepackage{bm}
\usepackage{caption}
\usepackage{amsmath, amssymb,amsthm}
\usepackage{upgreek}
\usepackage{color}
\usepackage{float}
\usepackage{xfrac}
\usepackage{stfloats}

\newtheorem{proposition}{Proposition}

\newtheorem{corollary}{Corollary}

\ifCLASSINFOpdf
\graphicspath{{pic/}}
\else
\fi
\begin{document}
%

\title{Throughput Analysis  in Cache-enabled\\  Millimeter Wave HetNets with\\ Access and Backhaul Integration}
%

\author{\IEEEauthorblockN{Chenwu~Zhang,
        Hao~Wu,
        Hancheng~Lu,
        Jinxue~Liu
        }
\IEEEauthorblockA{University of Science and Technology of China, Hefei, China, 230027\\
Email: cwzhang@mail.ustc.edu.cn, hwu2014@mail.ustc.edu.cn,
 hclu@ustc.edu.cn, jxliu18@mail.ustc.edu.cn}

}

\maketitle

\begin{abstract}
Recently, a mmWave-based access and backhaul integration heterogeneous networks (HetNets) architecture (mABHetNets) has been envisioned to provide high wireless capacity. Since the access link and the backhaul link  share the same mm-wave spectral resource, a large spectrum bandwidth is occupied by the backhaul link, which hinders the  wireless access capacity improvement. To overcome the backhaul spectrum occupation problem  and improve the network throughput in the existing mABHetNets, we introduce the cache at base stations (BSs). In detail, by caching popular files at small base stations (SBSs), mABHetNets can of\mbox{}f\mbox{}load  the backhaul link traffic and transfer the redundant   backhaul spectrum to the access link to increase the network throughout. However, introducing cache in SBSs will also incur additional power consumption and reduce the transmission power, which can lower the network throughput. In this paper, we investigate  spectrum partition between the access link and the backhaul link as well as cache  allocation to improve the network throughput in mABHetNets. With the stochastic geometry tool, we develop an  analytical framework to characterize cache-enabled mABHetNets and  obtain  the signal-to-interference-plus-noise ratio (SINR) distributions in line-of-sight (LoS) and non-line-of-sight (NLoS) paths.   Then we  utilize the SINR distribution to derive the average potential throughput (APT).  Extensive numerical  results show that introducing cache can bring up to 80\% APT gain  to the existing mABHetNets.
\end{abstract}


%
\IEEEpeerreviewmaketitle

\section{Introduction}
In recent years,   mmWave-based access and backhaul integration heterogeneous cellular networks (mABHetNets) has been envisioned in the 5th generation mobile communication technology (5G) dense cellular networks to satisfy the rapidly growing traffic demand \cite{IAB1,IAB2,3GPPIAB}. In mABHetNets, high-power mmWave micro base stations (MBSs) are overlaid by denser  lower-power mmWave small base stations (SBSs) where   MBSs and SBSs provide high  rate service to users by wireless access link, while the MBSs  maintain the backhaul capacity of the SBSs by wireless backhaul link. 
In 5G dense cellular networks,  average potential throughput (APT) is a key performance metric \cite{Throuhgput1,Throuhgput2}. 
 Given that both the access link and the backhaul link  share the same mmWave spectrum resource, mmWave spectrum partition between access link and backhaul link has played an important role in APT \cite{Partition1,Partition2,Partition3}.  \cite{Partition1}  explores the optimal partition of access and backhaul spectrum to maximize the rate coverage.  \cite{Partition2} leverages allocated resource ratio between radio access and backhaul to study maximization of network capacity by considering the fairness among SBSs. \cite{Partition3} jointly studies the beamforming and spectrum partition to improve the network capacity of mABHetNets. However,   according to the findings in \cite{Partition3}, up to 50\% mmWave spectrum is used in backhaul link to satisfy the high speed data traffic.  Such stubborn ``\emph{spectrum occupation}'' in mABHetNets has  restricted network performance such as APT to achieve a better possible improvement.

Nowadays, enabling caching at base stations (BSs)   has been considered as a promising way to alleviate the backhaul   problem \cite{Cacheaware,2u,caching,MostPop}. Statistical reports have shown that a few popular files requested by many users account for most of the backhaul traffic load \cite{2u,MostPop}. Based on this fact, popular files can be proactively cached at SBSs, and then the cached files can be delivered to users from SBSs directly, which can significantly alleviate backhaul traffic pressure. \cite{Cacheaware} exploits the cache 
to overcome the backhaul capacity limitations and enhance  users' quality of service (QoS). Besides, \cite{caching} considers the most popular file caching strategy to minimize the backhaul network cost.

However, most existing work mainly focuses on caching in wire backhaul based cellular networks, while neglecting the role of caching in wireless backhaul based  mABHetNets. In fact, by caching popular files at the cache of SBSs,  mABHetNets can of\mbox{}f\mbox{}load a lot of the wireless backhaul traffic. With the backhaul traffic reduced, a part of backhaul mmWave spectrum  can be transferred to the access link, which further increases the access link capacity. Since the cache will consume  the power resource,  the  transmission power will be reduced, which will lower the data rate\cite{BSPowerModel,BSPowerModel1}. Therefore, power consumption should be carefully considered when cache capability is enabled at SBSs.

In this paper,  we introduce the  cache in SBSs to  improve the network throughput of mABHetNets. To the best of our knowledge, there is no theoretical research that investigates the performance of the mABHetNets when cache is involved. Motivated by such fact, we propose  a stochastic geometry  based analytical framework for the two-tier cache-enabled  mABHetNets.  In this framework,  we derive the APT expression. Then we investigate the joint impact  of spectrum  partition   and   cache allocation on APT.  The major contributions of this paper are summarized as follows.
\begin{itemize}
\item By stochastic geometry tool, we develop a  fundamental analytical framework to study the cache-enabled mABHetNets, where SBS caching model and spectrum partition model are carefully described.
\item Based on  the analytical framework, we analyze  SBS and MBS association probability ,and obtain  the SINR distributions of the the line-of-sight (LoS) and non-line-of-sight (NLoS) transmission path. Finally,  we  derive the APT expressions  of users covered by the SBS tier and MBS tier, respectively, where  spectrum  partition and  cache  allocation are involved.
\item Extensive numerical  results are carried out to  validate the motivation and effectiveness of our work. We depict the impact of spectrum  partition and cache  allocation on APT. Then we found that there exist optimal spectrum partition and cache allocation. Besides, comparing with no cache in mABHetNets, introducing cache can add 80\% APT gain.
\end{itemize}

The rest of the paper is organized as follows. Sect. II gives an overview of the system model. APT is derived and analyzed    in Sect. III. Next, numerical results are given in Sect. IV. Finally, the conclusions of our work are drawn in Sect. V.

\section{System model}

%
In this section, 
  a two-tier downlink mABHetNet is considered as shown in Fig. \ref{example}. The first tier  consists of lower-power SBSs while the second tier consists of higher-power MBSs. The MBSs are connected to the mobile core network by high-capacity  optical fiber links. Besides, the MBSs provide wireless backhaul connections to the SBSs via   broad mmWave spectrum. Both MBSs and SBSs provide the access service to users. The locations of  MBSs  and SBSs are assumed to follow independent Poisson point processes (PPPs), which are denoted by $\Phi_m\in \mathbb{R}_2$ and $\Phi_s\in \mathbb{R}_2$  with densities  $\lambda_m$ and $\lambda_s$.  The user density  $\lambda$
is assumed to be sufficiently large  so that each BS has at least one associated user in its coverage area.

\begin{figure}[htbp]
  \centering
  \includegraphics[width=2.6in]{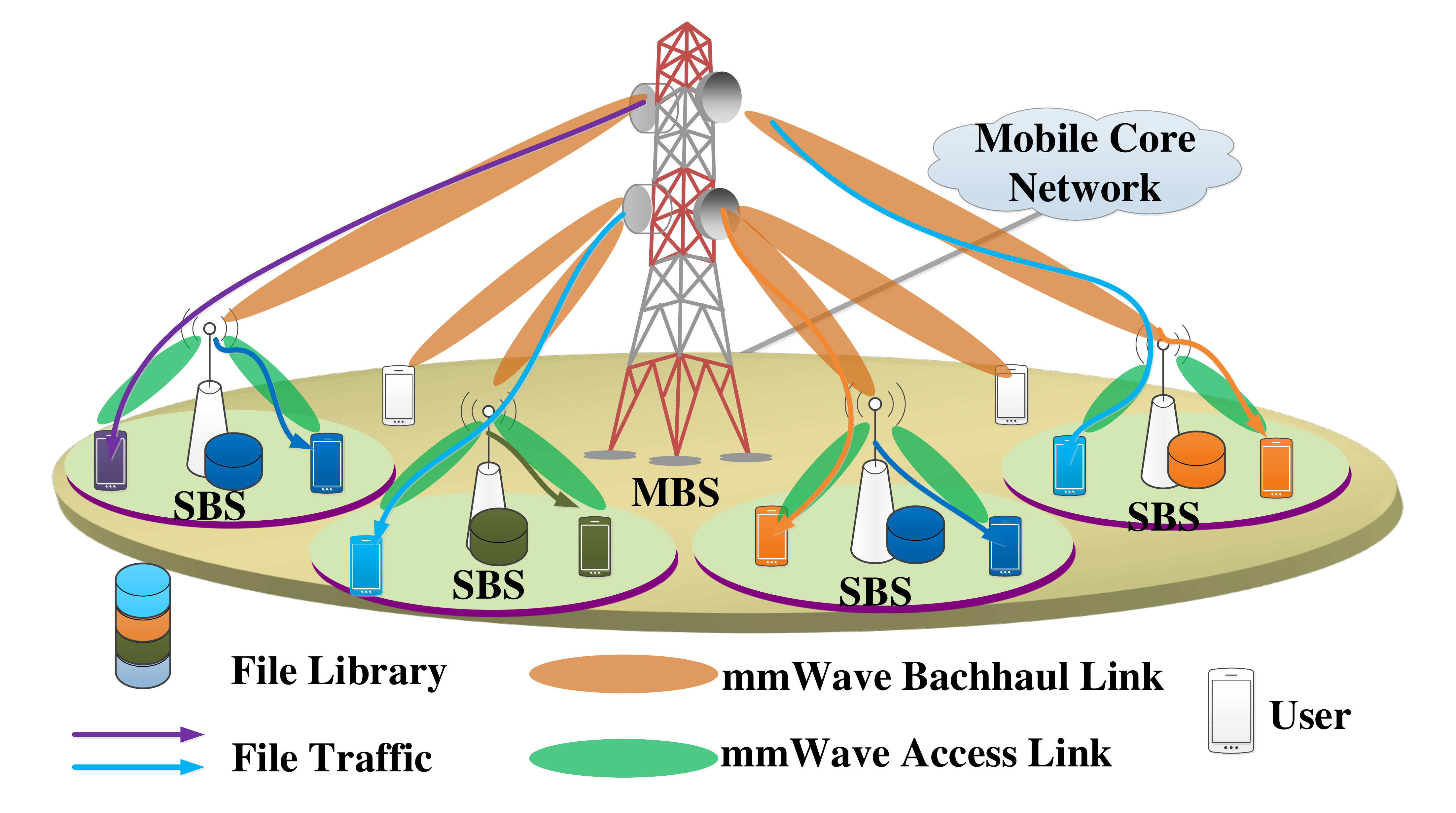}
  \captionsetup{font={small}}
  \caption{Example of downlink mABHetNet with two tiers of BSs:   cache-enabled MBSs  are overlaid with denser cache-enabled SBSs}
  \label{example}
\end{figure}


\begin{figure*}[b]
\hrulefill 
\begin{align} \label{LOSModel}
\qquad \qquad\operatorname{SINR}_{s}(r)
=\frac{P_{s}^{tr} B_{s}h_{s  } L(r)}{I_s+I_m+N_{0}}
 =\nonumber
 \frac{(P'_s-{{w'}_{\mathrm{ca}}^{s}}C) B_{s}h_{s  } L(r_{s })}
   {\sum\limits_{i \in \Phi_{s} \backslash b_{s,0}}(P'_s-{{w'}_{\mathrm{ca}}^{s}}C) B_{s } h_{s, i}L(r_{s, i })
   +\sum\limits_{l \in \Phi_{m} } P_{m}^{tr}  B_{m}h_{m, l }L(r_{m, l })+N_{0}}\tag{3} 
   \\
\qquad \qquad \operatorname{SINR}_{m}(r)
=\frac{P_{m}^{tr} B_{m}h_m L(r_{m} )}
{I'_s+I'_m+N_{0}}=\nonumber
\frac{P_{m}^{tr} B_{m}h_m L(r_{m} )}
{\sum\limits_{i \in \Phi_{s}}  (P'_s-{{w'}_{\mathrm{ca}}^{s}}C)B_{s} h_{s, i} L(r_{s, i})+\sum\limits_{l \in \Phi_{m} \backslash b_{m,0}} P_{m}^{tr} B_{m} h_{m, l} L(r_{m, l})+N_{0}}
\tag{4}
\end{align}
\end{figure*}

\begin{table}
\centering
\caption{Main Symbols}
\label{Symbols}
\renewcommand\arraystretch{1.2}
\begin{tabular}{cp{6cm}}
\toprule
\textbf{Symbol}                     & \textbf{Meaning}                                       \\
\midrule
$W,W_{ac},W_{bh}$                         &Total spectrum bandwidth, access link bandwidth and backhaul link bandwidth\\
$\eta$                              &mmWave spectrum partition ratio for access link\\
$\alpha_\mathrm{L},\alpha_{\mathrm{NL}}$                            & Path loss exponet in LoS and NLoS transmission\\
$\lambda_m,\lambda_s,\lambda$                             &Density of MBS, SBS and user, respectively\\
$C$                                 &Cache capacity of an SBS\\
$F$                                 &Number of files \\
$p_h$                               &Cache hit ratio of an SBS\\
$P_{m}^{tot},P_{s}^{tot}$           &Total power of an MBS and an SBS, respectively\\
$P_{m}^{tr},P_{s}^{tr}$             &Transmission power of an MBS and an SBS,respectively\\
$B_m,B_s$                           &The association bias factor of MBS and SBS\\
$w_{ca}$                            &Caching power consumption coefficient\\
\bottomrule
\end{tabular}
\end{table}
\subsection{Caching Model}
The file library is denoted by $\mathcal{F}$ and there are $|\mathcal{F}|= F$ files in the library. It is assumed that each file has the equal size \cite{Samesize}.  Zipf distribution  is widely used to model the popularity of file $f$, $\forall f \in\{1,...,F\}$:
$
  p_{f}=\frac{f^{-\gamma_p}}{\sum_{g=1}^{F}g^{-\gamma_p}}
$,
where $\gamma_p$ is the Zipf exponent reflecting different levels of skewness of the distribution.  $\gamma_p$ is between 0.5 and 1.0, where higher value causes more ``peakiness'' of the distribution \cite{MostPop}.

The cache capacity of the SBS and MBS is denoted by $F$ and $C$ (file units), respectively. As MBS is often equipped with  large  cache capacity, we assume that the MBS can cache all the files\cite{MBSFileLibrary}.   Since each SBS has a limited cache capacity, it only can contain $C$ files.  Each SBS caches the
most popular contents and the cache hit ratio of an SBS is calculated as\cite{MostPop}
\begin{equation}\label{HitRatio}
  p_{h}=\frac{\sum_{f=1}^{C}f^{-\gamma_p}}{\sum_{g=1}^{F}g^{-\gamma_p}}
\end{equation}

\subsection{Wireless Transmission Model}
The mmWave based  wireless   link   can be either line of sight (LOS) or non-line of sight (NLoS) transmission. Then, the  path loss functions are defined as \cite{LOSModel}
\begin{align}\label{LOSModel}
 \operatorname{L}(r)
 =\left\{
 \begin{array}{ll}
 {A_{\mathrm{L}} r^{-\alpha_{\mathrm{L}}},}      &{\text{with  $\mathcal{P}_\mathrm{L}(r)$}} \\
 {A_{\mathrm{NL}} r^{-\alpha_{\mathrm{NL}}},}    &{\text{with  $\mathcal{P}_\mathrm{NL}(r)=1-\mathcal{P}_\mathrm{L}(r)$}}
 \end{array}
 \right.
\end{align}
together with the   LoS probability  $\mathcal{P}_{\mathrm{L}}(r)=\min\left(\frac{18}{r},1\right)(1-e^{-\beta r})+e^{-\beta r}$
and $r$ is the transmission distance.   $\beta \geq 0$ is the parameter that captures
density and size of obstacles between the transmitter and the reveicer.

By involving the caching power consumption,     the power consumption model   of  one SBS is \cite{BSPowerModel}:$P_{s}^{tot}= \rho_s P_{s}^{tr}+P_{s}^{fc}+P_{s}^{ca}$,
where   $P_{s}^{tr}$ denote transmit powers consumed at  an SBS. $ \rho_s$ reflects the impact of power amplifier and cooling on transmit power. $P_{s}^{\mathrm{fc}}$ is the fixed circuits-related power consumption. The caching power  consumption is   $P_{s}^{ca}=w_{ca}C$, where $\omega_{ca}$ is the power coefficient of cache hardware in watt/bit. For each SBS, since the total power consumption $P_{s}^{tot}$ is usually given,    $P_s^{tr}=\frac{P_{s}^{tot}-P_{\mathrm{s}}^{\mathrm{fc}}-w_{\mathrm{ca}}C}{\rho_s}
=P'_s-{w'}_{\mathrm{ca}}^{s}C$, where $P'_s=\frac{P_{s}^{tot}-P_{\mathrm{s}}^{\mathrm{fc}}}{\rho_s}$ and ${w'}_{ca}^s=\frac{w_{ca}}{\rho_s}$. Note that considering the total  power  consumption,  the maximum cache capacity  is  $\frac{P_{s}^{tot}}{w_{\mathrm{ca}}}$.
 Similarly,  the  transmission power of one MBS is $P_m^{tr}
=P'_m-{w'}_{\mathrm{ca}}^{m}F$, where $P'_m=\frac{P_{m}^{tot}-P_{\mathrm{m}}^{\mathrm{fc}}}{\rho_m}$ and ${w'}_{ca}^m=\frac{w_{ca}}{\rho_m}$.  Here, since the MBS contains all files, the  transmission power of one MBS is fixed.

The  SINR  of a typical user at a   distance $r$ from its associated SBS and MBS  are at the bottom of next page.
   $h_{s},h_{m }$ are the small-scale fadings from SBS and MBS. $L(r_{m}),L(r_{s})$ are the path losses   from the serving SBS or MBS to the typical user. $r_{s}$ ($r_{m}$) is the distance between the assocation SBS $b_{s,0}$  (assocation MBS $b_{m,0}$) and the typical user.  $ r_{s, i}$($ r_{m,l}$) is the distance between the $i$-th SBS($l$-th MBS )and the typical user.    $N_{0}$ is   the additive white Gaussian noise. To backhaul the data traffic of the SBS, the MBS provides the wireless backhaul link. For a typical SBS that a random distance $r_{bh}$ from its associated MBS, the SINR of   downlink backhaul is
\begin{align}\label{M-SBSSINR}
&\operatorname{SINR}_{bh} (r_{bh})
=\frac{P_{m}^{tr} B_m h_{m } L(r_{bh})}
{ I_{bh}+  N_0} \nonumber  \\&=
\frac{P_{m}^{tr}B_m h_{m } L(r_{bh})}
{ \sum_{i \in \Phi_{m} \backslash b_{m,0 }} P_{m}^{tr} B_m  h_{m,i}L(r_{bh,i})+  N_0}
\setcounter{equation}{4}
\end{align}

\subsection{Spectrum Partition Model}
\vspace{-0.3cm}
\begin{figure}[htbp]
  \centering
  \setlength{\abovecaptionskip}{0.cm}
  \includegraphics[width=3.5in]{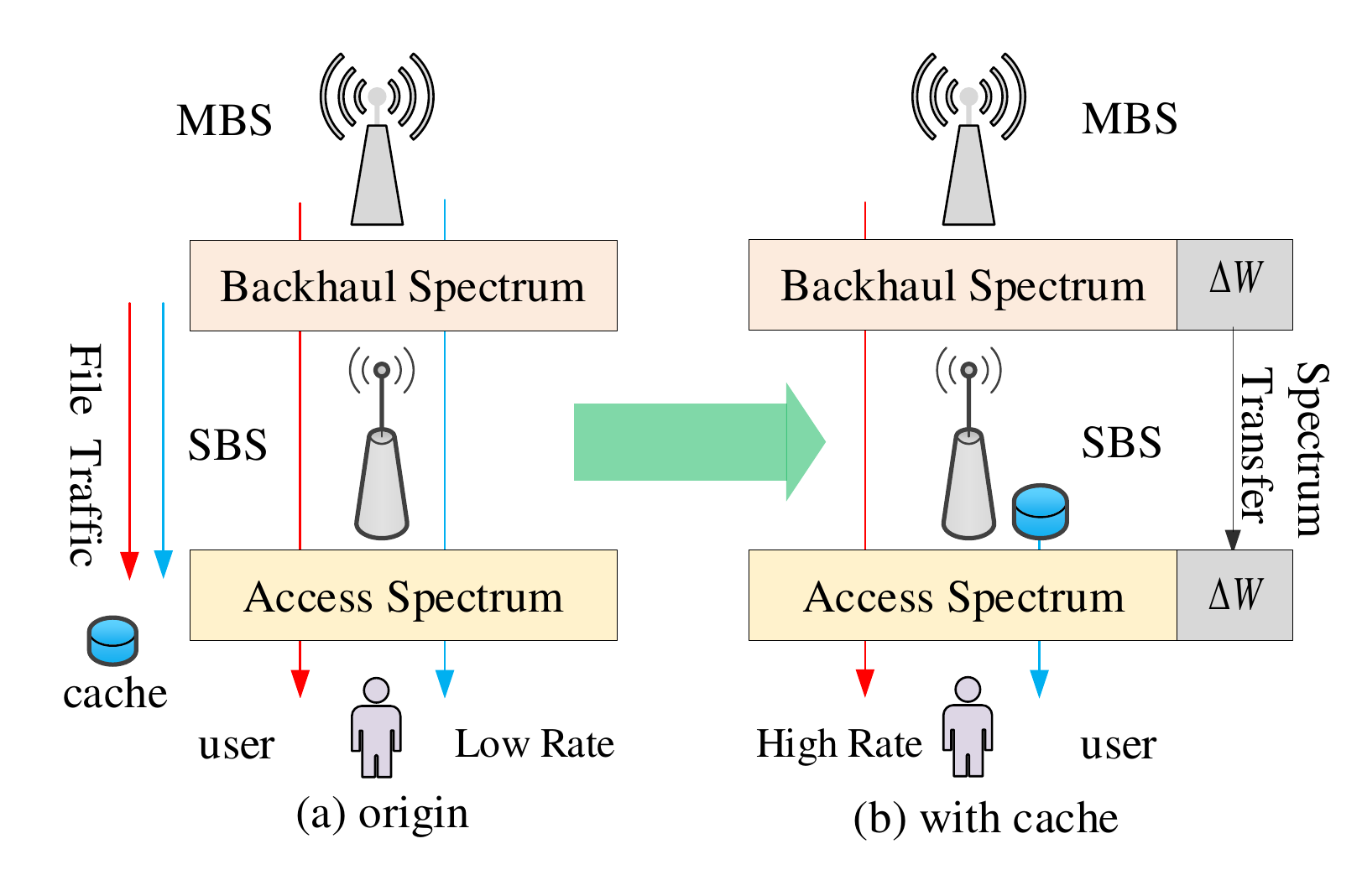}
  \captionsetup{font={small}}
  \caption{Spectrum partition between access and backhaul (a) original and (b)with cache}
  \label{schematic}
\end{figure}
 In Fig. \ref{schematic}, both the access link and the backhaul link share the mmWave spectrum.  The total mmWave spectrum bandwidth $W$ for downlink transmission is partitioned into two parts: $W_{ac}=\eta W$  for access and $W_{bh}=(1-\eta )W$ for backhaul.  $\eta \in [0, 1]$  denotes the part of spectrum for access link. The file delivery is both related with the access link and backhaul link. Then  by changing $\eta $, both the access rate and backhaul rate in the same transmission path can maintain an effective  transmission. When the cache is introduced, cached files can be directly delivered from SBS and the backhaul traffic can be saved (Fig. \ref{schematic}(b) ). Thus, part mmWave spectrum resource in backhaul link can be shifted to access link and the data rate is improved.

\section{APT analysis}

In this section, we analyze and derive APT in the cache-enabled mABHetNet. Firstly, we derive  the  SINR distribution of the typical user. Then, the SINR coverage probability of the SBS covered by the MBS is also obtained. Finally, APT is derived based on analysis of the SINR.

\subsection{SINR analysis}
To study the APT  performance of the mABHetNet, we need to investigate the SINR distribution of the user  covered by SBS/MBS tier via access link or the SINR distribution of  the SBS   covered by MBS via backhaul link.
\subsubsection{Association Probability}
According to \cite{ProofProcedure}, the probability distribution functions (PDFs) of $r$ (the distance between the  user and the  nearest BS via a LoS/NLoS path) are
    \begin{align}\label{LOSUserClosestSBSProbabilty}
        f_{R_{k}}^{\mathrm{L}}(r)     &=\mathcal{P}_L(r) \times \exp \left(-\pi r^{2} \lambda_k\right) \times 2 \pi r \lambda_k \\
        f_{R_{k}}^{\mathrm{N L}}(r)   &= \mathcal{P}_\mathrm{NL}(r) \times \exp \left(-\pi r^{2} \lambda_k\right) \times 2 \pi r \lambda_k
    \end{align}
where $k\in \{s,m\}$ denotes the index of SBS tier or MBS tier.

 Similarly,  the PDFs of distance $r$ (between  the SBS  and the  associated the nearest MBS via a LoS/NLoS path)  are
\begin{align}\label{PDFSBSMBSLOS}
  f_{R_{bh}}^{\mathrm{L}}(r)&=\mathcal{P}_L(r) \times  \exp \left(-\pi r^{2} \lambda_m\right) \times 2 \pi r \lambda_m \\
  f_{R_{bh}}^{\mathrm{NL}}(r)&=\mathcal{P}_{\mathrm{NL}}(r) \times  \exp \left(-\pi r^{2} \lambda_m\right) \times 2 \pi r \lambda_m
\end{align}

In the mABHetNet,   we need to analyze the probability that a user  is associated with SBS tier or with MBS tier.  Besides, since  SBS may be associated with MBS via different  transmission paths, different SBS backhaul association probabilities should be derived.
We consider the  maximum biased received power where a mobile user is associated with the strongest BS in terms of  the received power.  The following lemma provides the per-tier association probability via LoS and NLoS path.
\begin{proposition}\label{SBSMBSAsso}
For the given distance $r$, the probabilities that a typical user is associated with the SBS tier by LoS link and NLoS link are
    \begin{align}\label{UserAssociatedSBS}
    F_{s}^{\mathrm{L}}(r) &=p_{ln}^{ss}(r)p_{ll}^{sm}(r)p_{ln}^{sm}(r)f_{R_{s}}^{\mathrm{L}}(r)\\
    F_{s}^{\mathrm{NL}}(r)&=p_{nl}^{ss}(r)p_{nl}^{sm}(r)p_{nn}^{sm}(r)f_{R_{s}}^{\mathrm{NL}}(r)
\end{align}
Then, the probabilities that a typical user is associated with an MBS tier by LoS link and NLoS link are
\begin{align}\label{UserAssociatedMBS}
  F_{m}^{\mathrm{L}}(r)&=p_{ln}^{mm}(r)p_{ll}^{ms}(r)p_{ln}^{ms}(r)f_{R_{m}}^{\mathrm{L}}(r)\\
  F_{m}^{\mathrm{NL}}(r)&=p_{nl}^{mm}(r)p_{nl}^{ms}(r)p_{nn}^{ms}(r)f_{R_{m}}^{\mathrm{NL}}(r)
\end{align}
where  $p_{ln}^{ss}(r)$ denotes the probability of the event that the user  obtains the desired LoS signal from SBS tier and the NLoS  interference from an SBS tier. The other probabilities have the similar definitions and can be found in the proof.
\end{proposition}

\begin{proof}
   Due to the limited  space, we omit the proof.  The detailed proof procedure can be found in Sect. VI-B in our technical report \cite{ProofProcedure}.
\end{proof}

\begin{corollary}
 Similar to Proposition \ref{SBSMBSAsso}, the probabilities that a typical SBS is associated with   an MBS tier by LoS link and NLoS
link are
\begin{align}\label{SBSAssociatedMBS}
&F_{bh}^{L}(r)=p_{ln}^{bh}(r) f_{R_bh}^{\mathrm{L}}(r)\\
&F_{bh}^{NL}(r)=p_{nl}^{bh}(r) f_{R_bh}^{\mathrm{NL}}(r)
\end{align}
where $p_{ln}^{bh}(r)$ is the probability that an SBS is associated with a LoS MBS and
the interference is from NLoS MBS.  $p_{nl}^{bh}(r)$ has the similar definition. The details of $p_{ln}^{bh}(r)$ and $p_{nl}^{bh}(r)$  can be found in Sect. VI-B in our technical report \cite{ProofProcedure}.
\end{corollary}

\subsubsection{SINR Distribution}
This SINR distribution is defined as the SINR coverage probability that the received SINR is above a pre-designated threshold $\gamma$:
\begin{equation}
P^{\operatorname{cov}}(  \gamma)=\operatorname{Pr}[\operatorname{SINR}>\gamma]
\end{equation}
Since the user is covered either by an SBS tier or an MBS tier, we give the two SINR distributions. Then we give the SINR distribution of a typical SBS when it is covered by an MBS.
\begin{proposition}\label{ProUserSINRCoverageProbability}
The SINR coverage probability that the user is associated with the SBS  is
\begin{align}\label{UserSBSSINRCoverageProbability}
&\mathbb{P}_{k}^{cov}(\gamma)=P_{k,L}^{cov}(\gamma)+P_{k,NL}^{cov}(\gamma)\nonumber \\
 P_{k,L}^{cov}(\gamma)&=\int_{0}^{\infty}  \exp \left(\frac{-\gamma N_{0}r^{\alpha_{\mathrm{L}}}}{P_{k}^{tr}B_kA_{\mathrm{L}} } \right)  \mathcal{L}_{I_{k}}^{\mathrm{L}}  F_{k}^{\mathrm{L}}(r)\mathrm{d}r\nonumber\\
P_{k,NL}^{cov}(\gamma)&=\int_{0}^{\infty}  \exp \left(\frac{-\gamma N_{0} r^{\alpha_{\mathrm{NL}}}}{P_{k}^{tr}B_kA_{\mathrm{NL}} }\right)  \mathcal{L}_{I_{k}}^{\mathrm{NL}}   F_{k}^{\mathrm{NL}}(r)d r
\end{align}
where $k\in\{s,m\}$ denotes SBS or MBS.  $P_{k,L}^{cov}(\gamma)=\mathbb{E}_{r }\left[\mathbb{P}\left[\mathrm{SINR}_{k}^{\mathrm{L}} (r ) \geq \gamma\right]\right]$ is the probability that the user is covered by SBS or MBS  with LoS based signal
and $P_{k,NL}^{cov}(\gamma)=\mathbb{E}_{r }\left[\mathbb{P}\left[\mathrm{SINR}_{k}^{\mathrm{NL}} (r ) \geq \gamma\right]\right]$ is the similar probability with NLoS based signal. $\gamma$ is the threshold for successful demodulation and decoding at the receiver. Besides, $\mathcal{L}_{I_{s}}^{\mathrm{L}} = \mathcal{L}_{I_{s,m}}^{\mathrm{L}}\left(\gamma r^{\alpha_L}\right)$,$\mathcal{L}_{I_{s}}^{\mathrm{NL}} =\mathcal{L}_{I_{s,m}}^{\mathrm{NL}}\left(\gamma r^{\alpha_{NL}}\right)$, $\mathcal{L}_{I_{m}}^{\mathrm{L}}= \mathcal{L}_{I'_{s,m}}^{\mathrm{L}}\left( \gamma r^{\alpha_L}\right)$ and $\mathcal{L}_{I_{m}}^{\mathrm{NL}}= \mathcal{L}_{I'_{s,m}}^{\mathrm{NL}}\left( \gamma r^{\alpha_L}\right)$.
\end{proposition}
\begin{proof}
The detailed proof procedure can be seen in Appendix A.
\end{proof}
Similarly, the SINR coverage probability that an SBS is covered by an  MBS via wireless backhaul is
\begin{small}
\begin{align}\label{BHSINRCoverageProbability}
  \mathbb{P}_{bh}^{cov}(\gamma) &  =P_{bh,\mathrm{L}}^{cov}(\gamma)+P_{bh,\mathrm{NL}}^{cov}(\gamma)\nonumber \\
  P_{bh,\mathrm{L}}^{cov}(\gamma) & =\int_{0}^{\infty} \exp \left(\frac{-\gamma N_{0}r^{\alpha_{\mathrm{L}}}}{P_{m}^{tr}  B_m A_{\mathrm{L}} } \right)  \mathcal{L}_{I_{bh}}^{\mathrm{L}}\left(\gamma r^{-\alpha_L}\right) F_{bh}^{\mathrm{L}}(r)dr\nonumber\\
  P_{bh,\mathrm{NL}}^{cov}(\gamma) & =\int_{0}^{\infty}\exp \left(\frac{-\gamma N_{0} r^{\alpha_{\mathrm{L}}}}{P_{m}^{tr}  B_mA_{\mathrm{NL}} } \right)  \mathcal{L}_{I_{bh}}^{\mathrm{NL}}\left(\gamma r^{-\alpha_{NL}}\right) F_{bh}^{\mathrm{NL}}(r) d r
\end{align}
\end{small}

%


\subsection{ APT of   Cache-enable mABHetNet}
APT is a significant metric to measure the network performance and it focuses on the average user QoS requirement in terms of data rate. Next, we derive the APT expression and analyze it.
APT captures the average number of bits that can be received by the user per unit area   per unit bandwidth given a pre-designated threshold $\gamma$ \cite{Throuhgput1}. The definition of APT  is
\begin{equation}
   \mathcal{R}\left(  \gamma\right) =\lambda_{bs} W \log _{2}\left(1+\gamma\right) \mathbb{P}\left\{\operatorname{SINR} \geq \gamma\right\}
\end{equation}
where $\lambda_{bs}$ is the density of BS. $W $ is allocated bandwidth to the user. $\gamma$ is the user's SINR requirement.

In the cache-enable mABHetNet, APT is determined by cache capacity, spectrum partition and SINR threshold. Then
APT  of the cache-enable mABHetNet is denoted by
\begin{equation}\label{NetworkThroughput}
  \mathcal{R}(\eta,C,\gamma_0) = \mathcal{R}_s(\eta,C,\gamma_0)+\mathcal{R}_m(\eta,\gamma_0)
\end{equation}
where $\mathcal{R}_s$ and $\mathcal{R}_m$ are   APT of  an SBS tier and an MBS tier. $\gamma_0$ is the SINR threshold to guarantee the user throughout requirement. Then we will give the detailed APT expression.

\subsubsection{APT of user covered by an SBS tier}
For a user associated with an SBS, the transmission path includes the  access link between SBS and user, the backhaul link bewteen  SBS and MBS. Besides,  the cache in SBS tier also influence the file delivery in the transmission path. When the files are cached at the SBS, the files can be delivered to user directly. At this time, the wireless backhaul between the SBS and the MBS will not be used. Otherwise, the uncached files will be delivered through the wireless backhaul link. Given the cache in SBS,   we first give the APT of an SBS tier.
\begin{proposition}\label{proposAPTSBS}
Since the transmission  can be LoS or NLoS in wireless access link and wireless backhaul link for a user associated with an SBS tier, there are four cases in the SBS-tier throughput. Then
\begin{equation}\label{APTSBS}
   \mathcal{R}_s(\eta,C,\gamma_0) = \mathcal{R}_{s}^{ll}+\mathcal{R}_{s}^{ln}+\mathcal{R}_{s}^{nl}+\mathcal{R}_{s}^{nn}
\end{equation}
where $\mathcal{R}_{s}^{ll}$denotes the network throughput when the wireless SBS link and the wireless backhaul link, and $\mathcal{R}_{s}^{ln},\mathcal{R}_{s}^{nl},\mathcal{R}_{s}^{nn}$ have similar definitions.
\end{proposition}
\begin{proof}
The detailed proof procedure is given in the Appendix B.
\end{proof}

\subsubsection{APT of user covered by the MBS tier}
Since the signal is directly transmitted by MBS to user via access link, similar to  Proposition \ref{proposAPTSBS}, we can easily give the expression of APT of MBS tier.
\begin{corollary}\label{corolAPTMBS}
It is easy to obtain the average throughput of the MBS tier:
    \begin{align}\label{Throughput}
    \mathcal{R}_m(\eta,\gamma_0)
       =&\lambda_m\eta W\log_2(1+\gamma_0)P^{cov}_{m,\mathrm{L}}(\gamma_0) \nonumber   \\
       &+ \lambda_m\eta W\log_2(1+\gamma_0)P^{cov}_{m,\mathrm{NL}}(\gamma_0)
    \end{align}
where $P^{cov}_{m,\mathrm{L}}(\gamma_0)$ and $P^{cov}_{m,\mathrm{NL}}(\gamma_0)$ are the SINR coverage probability that the user is associated with the MBS via LoS and NLoS path in Proposition \ref{ProUserSINRCoverageProbability}.
\end{corollary}



\section{Numerical results}
In this section, simulations are performed to validate and evaluate the performance of the cache-enabled mHetNet. Particularly, we show the results of APT under different  scenarios.

\subsection{Parameter Setting}
%
$\lambda_s$ and  $\lambda_m$ are set to $10^{-4}$BSs/m$^2$  and $10^{-5}$BSs/m$^2$. The density of the users $\lambda$ is assumed to be $3\times10^{-4}$users/m$^2$.  $\gamma_p$   is set to 0.6\cite{MostPop}. Based on   \cite{BSPowerModel1},   each file unit has the same size of 4MB. The number of files in the file library is 1000 file units. The cache capacity of SBS    is 100 file units. To reflect the caching power model, we adopt the caching power coefficient $\omega_{ca}$ which is $2.5\times10^{-9}$W/bit \cite{BSPowerModel} and the total power of SBS and MBS is set as 9.1W and 610W to maintain the  transmission power consumption and caching power consumption. Other default simulation configurations are listed in Table \ref{parameters} \cite{Partition3}. All above settings will be changed according to different scenarios.

\renewcommand
\arraystretch{1.2}
\begin{table}[h]
\centering
\caption{Simulation parameters}
\label{parameters}
 \begin{tabular}{cc}
 \toprule
 Parameters                                  &  Values\\
 \midrule
 Total mmWave spectrum
 bandwidth $W$                        &400 MHz\\
 LoS pathloss parameters
  $A_\mathrm{L},\alpha_\mathrm{L}$         &$10^{-10.38},~2.09$\\
 NLoS pathloss parameters
 $A_{\mathrm{NL}},\alpha_{\mathrm{NL}}$  &$10^{-14.54},~3.75$\\
 Noise Power    $N_0$                                                &5 dB\\
 Fixed circuit power at MBS                                          &10.16W\\
 Fixed circuit power at SBS                                          &0.1W\\
 Power amplifier and
  cooling coefficient   $\rho_s$and $\rho_m$                  &4,~15.13\\
 Association biases of
 SBS and MBS $B_s$ and $B_m$                               &10,~1\\
 Blockage rate $\beta$                                               &2.7$\times$10$^{-2}$\\
 \bottomrule
 \end{tabular}
\end{table}
\subsection{Spectrum  Transfer under Different Cache Capacities}
\begin{figure}[H]
\centering
\includegraphics[width=2.3 in]{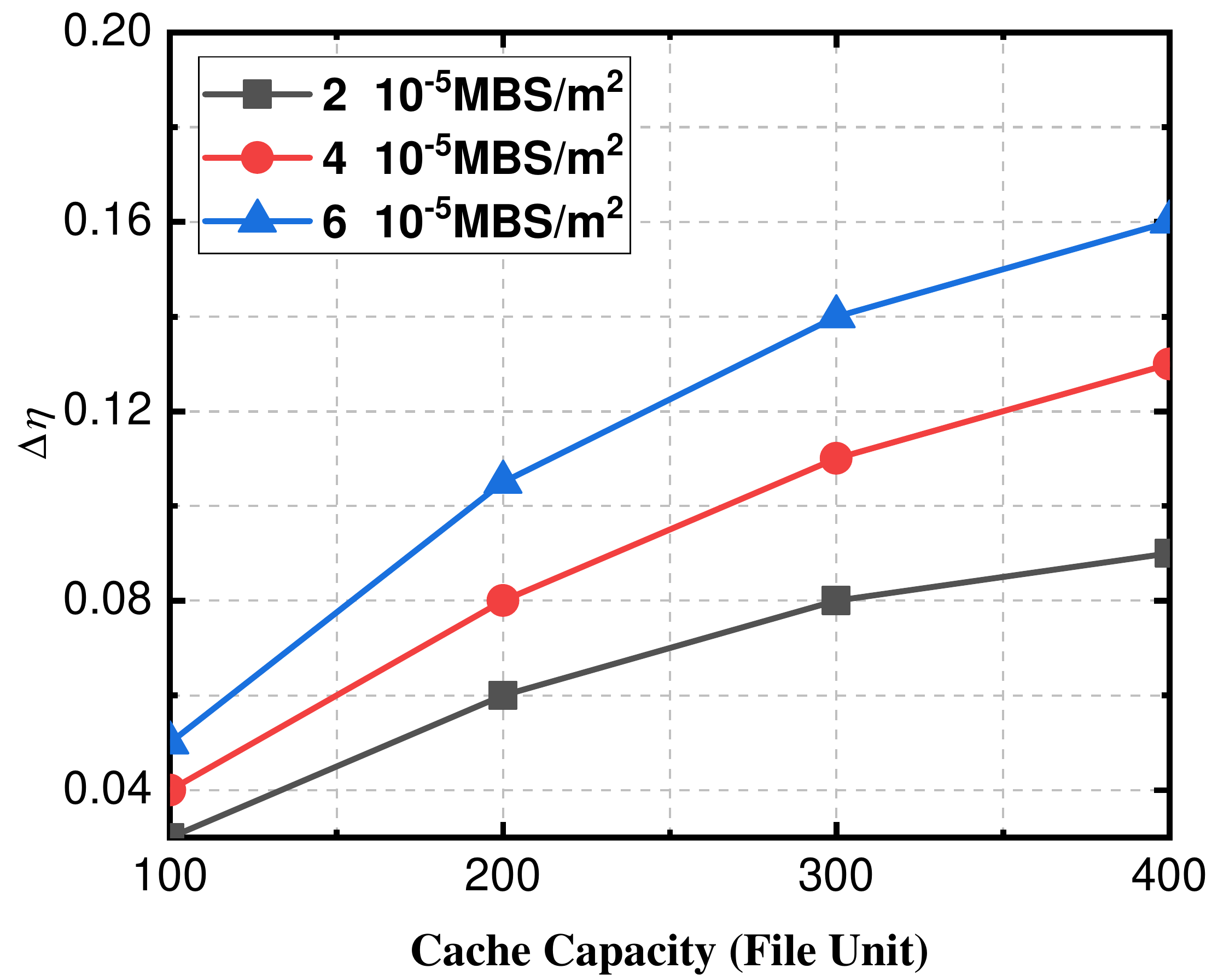}
\captionsetup{font={small}}
\caption{Spectrum  Transfer under different cache capacities( $\gamma_0$ is 10dB)}
\label{APTDelta}
\end{figure}
In Fig. \ref{APTDelta}, we show the spectrum transfer from the backhaul to  access link to obtain the optimal APT. When the cache capacity is changed from 100 to 400, more spectrum can be transferred from backhaul to access. This is because that as the cache capacity is increased, more files can be obtained from SBS directly. Then the backhaul traffic is reduced and the redundant spectrum  can be transferred to the access link to improve the data rate.  Besides, when the MBS tier becomes denser, the backhaul spectral efficiency is larger and the more spectrum can be transferred to  access link.

\subsection{The impact of Spectrum Partitions }
\begin{figure}[H]
\centering
\includegraphics[width=2.3 in]{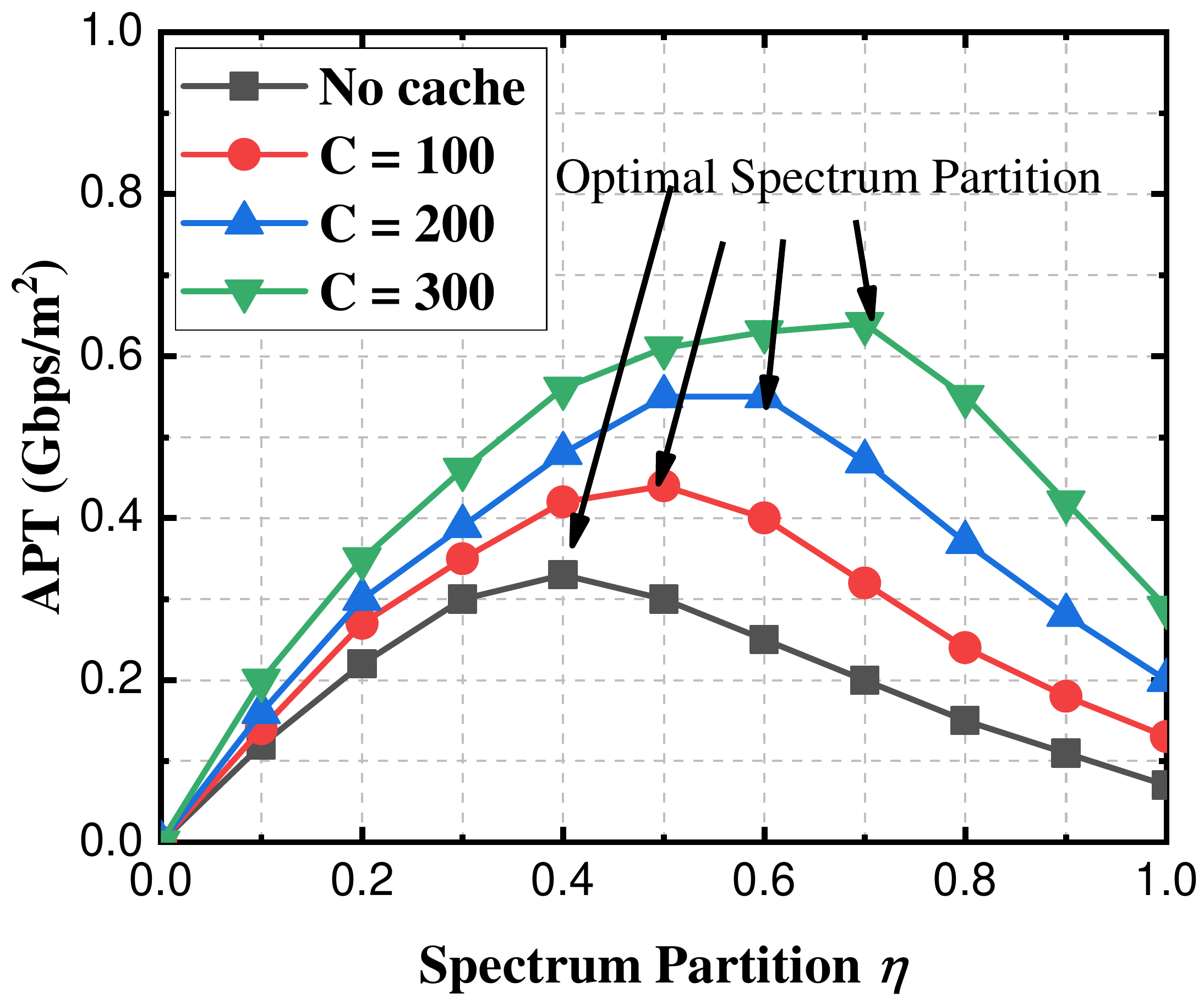}
\captionsetup{font={small}}
\caption{APT under different spectrum partitions( $\gamma_0$ is 10dB)}
\label{APTEta}
\end{figure}
In Fig. \ref{APTEta},   APT will increase as the spectrum partition increases and then it decreases. That means there exists the optimal spectrum in APT. This is because, when the backhaul spectrum bandwidth is enough, transferring some spectrum to the  access can increase the throughput of the user. Howerver, when more spectrum is used in access link, the backhaul link throughput cannot maintain the backhaul of the access throughput and the total throughput is reduced. We can also see that introducing the cache in the mABHetNet can bring 80\% APT gain approximately compared with  traditional mABHetNets with no cache.

\subsection{The impact of Cache Capacity}
\begin{figure}[H]
\centering
\includegraphics[width=2.3 in]{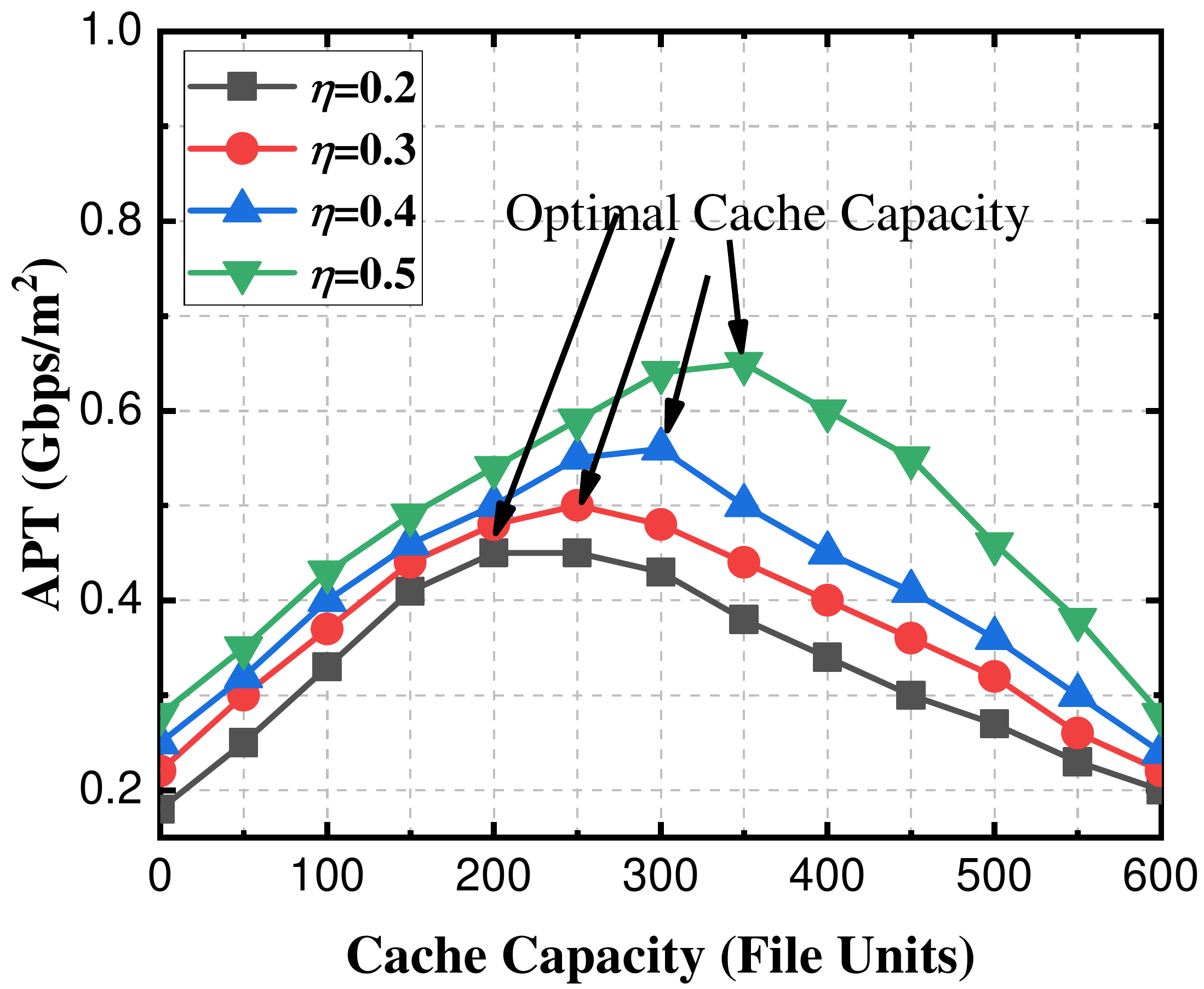}
\captionsetup{font={small}}
\caption{APT under different cache capacities($\gamma_0$ is 10dB)}
\label{APTCache}
\end{figure}

%
%
%

In Fig. \ref{APTCache}, APT also increases with the increasing cache capacity when the backhaul throughput  is limited with lower backhual spectrum. As more cache files can improve the cache hit ratio of SBS and   less files will  use backhaul resource, more files can be obtained by access without  spectrum. However, when the cache capacity is over 600, APT is close to zero. Such result is because that the maximum  power of SBS is limited and more cache capacity consumes more power and the transmission power is reduced. The reduced transmission power will decrease the data
 rate and APT.
\subsection{The impact of SINR Threshold }
\begin{figure}[H]
\centering
\includegraphics[width=2.3 in]{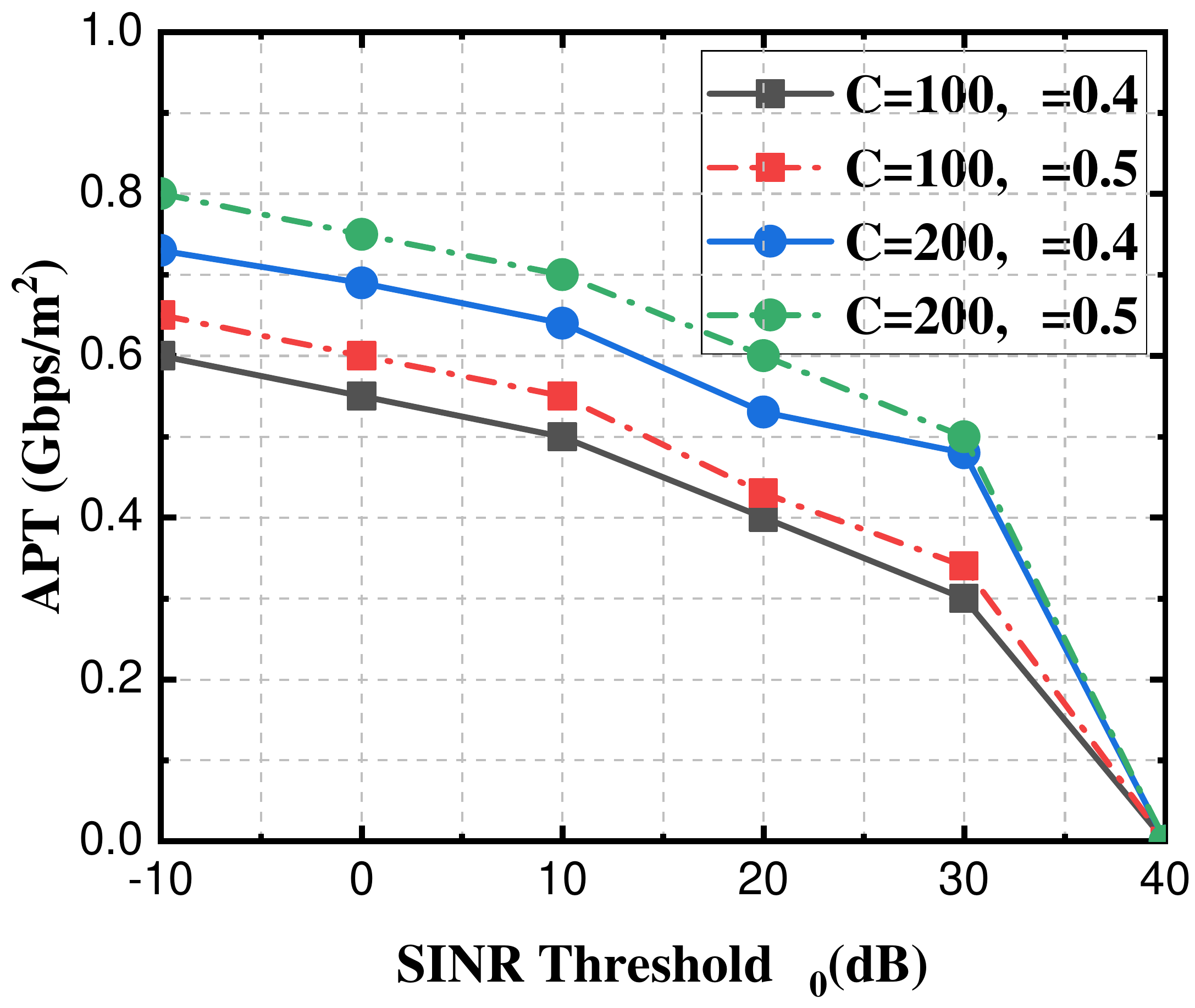}
\captionsetup{font={small}}
\caption{APT under different SINR thresholds}
\label{APTGamma}
\end{figure}
Since  APT is related to the user's rate requirement, APT is shown under different SINR thresholds in Fig. \ref{APTGamma}. High SINR threshold will decrease the APT of user. This is due to the fact that  the path loss and fading will make the received power lower and the received SINR lower than the threshold. However, at the same SINR threshold, more cache  and more spectrum partition can lead to more APT.
\section{Conclusion}
In this paper, we investigate  the impact of spectrum partition and cache capacity allocation on   APT in cache-enabled mABHetNets. Specifically, we designed a  cache-enabled mABHetNets model with the stochastic geometry tool and derived the SINR  distribution.  Then we used the SINR distribution  to obtain the   APT expression.  Via numerical
evaluations, we  first verified the impact of the  SBS cache on the spectrum partition and  found that there exist an  optimal spectrum partition and cache allocation to maximize APT. In the future, we will extend our work to multi-tier network scenarios.

\section{Appendix}
\subsection{Proof of Proposition \ref{ProUserSINRCoverageProbability}}\label{Appen3:ProUserSINRCoverageProbability}

The SINR distribution of  a user   covered by an SBS is
\begin{align}
  P_{s}^{cov}(\gamma)=P_{s,L}^{cov}(\gamma)+P_{s,NL}^{cov}(\gamma)
  \end{align}
where the SINR distribution of  a user   covered by a LoS  SBS is
\begin{small}
\begin{align}
  P_{s,L}^{cov}(\gamma)=&\mathbb{E}_{r }\left[\mathbb{P}\left[\mathrm{SINR}_{s}^{\mathrm{L}} (r ) \geq \gamma\right]\right]\nonumber \\
  =&\int_{0}^{\infty} \mathbb{P}\left[\operatorname{SINR}_{s}^{\mathrm{L}}(r)>\gamma\right] F_{s}^{\mathrm{L}}(r)\mathrm{d}r
\end{align}
\end{small}and the SINR distribution of  a user   covered by a NLoS  SBS is
\begin{small}
\begin{align}
  P_{s,NL}^{cov}(\gamma)=&\mathbb{E}_{r }\left[\mathbb{P}\left[\mathrm{SINR}_{s}^{\mathrm{NL}} (r ) \geq \gamma\right]\right] \nonumber \\
  =&\int_{0}^{\infty} \mathbb{P}\left[\operatorname{SINR}_{s}^{\mathrm{NL}}(r)>\gamma\right] F_{s}^{\mathrm{NL}}(r)\mathrm{d}r
\end{align}
\end{small}where $\gamma$ is the threshold for successful demodulation and decoding at the receiver. $\mathbb{P}\left[\mathrm{SINR}_{s}^{\mathrm{L}} (r ) \geq \gamma\right]$ means the probability of the event that the SINR of  the user covered by a SBS is over $\gamma$  via the LoS path at distance $r$:
\begin{align}
&\mathbb{P}\left[\operatorname{SINR}_{s}^{\mathrm{L}}(r)  \geq \gamma\right]\nonumber \\
&=\mathbb{P}\left[\frac{P_{s}^{tr}B_{s} A_{\mathrm{L}} r^{-\alpha_{\mathrm{L}}}}{I_s+I_m+N_{0}} \geq \gamma\right]\nonumber\\
&=\mathbb{P}\left[h_{m } \geq \frac{\gamma\left(I_s+I_m+N_{0}\right)}{P_{s}^{tr} B_{s}A_{\mathrm{L}} r^{-\alpha_{\mathrm{L}}} }\right]\nonumber \\ 
&\stackrel{(a)}{=} \exp \left(\frac{-\gamma N_{0}}{P_{s}^{tr} B_{s}A_{\mathrm{L}} r^{-\alpha_{\mathrm{L}}} }\right)  \mathcal{L}_{I_{s,m}}^{\mathrm{L}}\left(\gamma r^{\alpha_{\mathrm{L}}} \right)
\end{align}

Besides, $\mathbb{P}\left[\mathrm{SINR}_{s}^{\mathrm{NL}} (r ) \geq \gamma\right]$ means the probability of the event that the SINR of  the user covered by SBS is over $\gamma$  via the NLoS path at distance $r$:
\begin{align}
&\mathbb{P}\left[\operatorname{SINR}_{s}^{\mathrm{NL}}(r)  \geq \gamma\right] \nonumber \\
&=\mathbb{P}\left[\frac{P_{s}^{tr} B_{s} G_{s} A_{\mathrm{L}} r^{-\alpha_{\mathrm{NL}}}}{I_s+I_m+N_{0}} \geq \gamma\right]\nonumber\\
&=\mathbb{P}\left[h_{m0} \geq \frac{\gamma\left(I_s+I_m+N_{0}\right)}{P_{s}^{tr} B_{s}A_{\mathrm{NL}} r^{-\alpha_{\mathrm{NL}}} }\right]\nonumber\\
&\stackrel{(a)}{=} \exp \left(\frac{-\gamma N_{0}}{P_{s}^{tr} B_{s} h_{s}A_{\mathrm{NL}} r^{-\alpha_{\mathrm{NL}}} }\right)  \mathcal{L}_{I_{s,m}}^{\mathrm{NL}}\left(\gamma r^{\alpha_{\mathrm{NL}}} \right)
\end{align}
where (a) follows  small fading $h$$\sim$$\exp(1)$. Here the Rayleigh fading is considered.   Detailed laplace transform of the cumulative interference $\mathcal{L}_{I_{s,m}}^{\mathrm{L}}\left( \gamma r^{\alpha_{\mathrm{L}}} \right)$
$  \mathcal{L}_{I_{s,m}}^{\mathrm{NL}}\left(\gamma r^{\alpha_{\mathrm{NL}}} \right)$,
$  \mathcal{L}_{I'_{s,m}}^{\mathrm{L}}\left(\gamma r^{\alpha_{\mathrm{L}}} \right)$,
$  \mathcal{L}_{I'_{s,m}}^{\mathrm{NL}}\left(\gamma r^{\alpha_{\mathrm{NL}}} \right)$,
$  \mathcal{L}_{I_{bh}}^{\mathrm{L}}\left(\gamma r^{\alpha_{\mathrm{L}}} \right)$ and
$  \mathcal{L}_{I_{bh}}^{\mathrm{NL}}\left(\gamma r^{\alpha_{\mathrm{NL}}} \right)$
can    be found in Sect. VI-C in \cite{ProofProcedure}.

Next, we  focus on the SINR distribution of  a user   covered by an MBS :
\begin{align}
  P_{m}^{cov}
  =&P_{m,L}^{cov}(\gamma)+P_{m,NL}^{cov}(\gamma)  \nonumber \\
  =&\mathbb{E}_{r }\left[\mathbb{P}\left[\mathrm{SINR}_{m}^{\mathrm{L}} (r ) \geq \gamma\right]\right]
  +\mathbb{E}_{r }\left[\mathbb{P}\left[\mathrm{SINR}_{m}^{\mathrm{NL}} (r ) \geq \gamma\right]\right]\nonumber\\
  =&\int_{0}^{\infty} \mathbb{P}\left[\operatorname{SINR}_{m}^{\mathrm{L}}(r)>\gamma\right] F_{m}^{\mathrm{L}}(r) \mathrm{d} r \nonumber \\
  &+\int_{0}^{\infty}\mathbb{P}\left[\operatorname{SINR}_{m}^{\mathrm{NL}}(r)>\gamma\right] F_{m}^{\mathrm{NL}}(r) \mathrm{d} r
\end{align}
where 
\begin{align*}
\mathbb{P}\left[\operatorname{SINR}_{m}^{\mathrm{L}}(r)  \geq \gamma\right] 
= \exp \left(\frac{-\gamma N_{0}r^{\alpha_{\mathrm{L}}}}{P_{m}^{tr} B_{m}A_{\mathrm{L}}  }\right)  
\mathcal{L}_{I_{s,m}^{'}}^{\mathrm{L}}\left(\gamma r^{\alpha_{\mathrm{L}}} \right)
\end{align*}
\begin{align*}
\mathbb{P}\left[\operatorname{SINR}_{m}^{\mathrm{NL}}(r)  \geq \gamma\right]=  \exp \left(\frac{-\gamma N_{0} r^{\alpha_{\mathrm{NL}}}}{P_{m}^{tr} B_{m} g_{m}A_{\mathrm{NL} } }\right)  \mathcal{L}_{I_{s,m}^{'}}^{\mathrm{NL}}\left(\gamma r^{\alpha_{\mathrm{NL}}} \right)
\end{align*}

\subsection{The proof of Proposition \ref{proposAPTSBS}}
 Since the user data rate is related with the wireless access link  rate and the wireless backhaul link rate, when the access link and the backhaul link are both LoS, the user data rate is
  $\mathcal{R}_{s}^{ll}(\eta,C,\gamma_0)  = \min \{ \lambda_s\eta W\log_2(1+\gamma_0)P^{cov}_{s,\mathrm{L}}(\gamma_0), \frac{1 }{ 1-p_h }\lambda_m W(1-\eta)\log_2(1+\gamma_0)P^{cov}_{bh,\mathrm{L}}(\gamma_0)\}$. Following the same logic,
 \begin{align}\notag
    \mathcal{R}_{s}^{ln}    = &\min \{ \lambda_s\eta W\log_2(1+\gamma_0)P^{cov}_{s,\mathrm{L}}(\gamma_0),\\ \notag
    &\frac{1 }{ 1-p_h }\lambda_m W(1-\eta)\log_2(1+\gamma_0)P^{cov}_{bh,\mathrm{L}}(\gamma_0)\}\\\notag
    \mathcal{R}_{s}^{nl}  = &\min \{ \lambda_s\eta W\log_2(1+\gamma_0)P^{cov}_{s,\mathrm{L}}(\gamma_0),\\ \notag
    &\frac{1 }{ 1-p_h }\lambda_m W(1-\eta)\log_2(1+\gamma_0)P^{cov}_{bh,\mathrm{L}}(\gamma_0)\}\\\notag
    \mathcal{R}_{s}^{nn}    =& \min \{ \lambda_s\eta W\log_2(1+\gamma_0)P^{cov}_{s,\mathrm{L}}(\gamma_0),\\ \notag
    & \frac{1 }{ 1-p_h }\lambda_m W(1 -\eta)\log_2(1+\gamma_0)P^{cov}_{bh,\mathrm{L}}(\gamma_0)\}
\end{align}

Note that, $p_h = p_h(C) =\frac{\sum_{f=1}^{C}f^{-\gamma_p}}{\sum_{g=1}^{F}g^{-\gamma_p}}$ is the cache hit ratio in the SBS tier. $(1-p_h)$ reflects the probability   that the files not cached in SBS tier will be delivered through the backhaul link.  $P^{cov}_{s,\mathrm{L}}(\gamma_0)$ and $P^{cov}_{s,\mathrm{NL}}(\gamma_0)$ are the SINR coverage probabilities that the user is associated with the SBS via LoS and NLoS path in  Proposition \ref{ProUserSINRCoverageProbability}.


\end{document}